\documentclass{article}

\usepackage[T1]{fontenc}
\usepackage[utf8]{inputenc}
\usepackage{microtype}
\usepackage{geometry}

\usepackage{etex}
\usepackage{ifthen}

\usepackage{xspace}

\usepackage{amsmath}
\usepackage{amssymb}
\usepackage{amsthm}
\usepackage{bm}
\usepackage{mathtools}
\usepackage{thmtools}
\usepackage{thm-restate}

\usepackage{booktabs}
\usepackage{colortbl}
\usepackage{multirow}
\usepackage{multicol}
\usepackage{tabularx}

\usepackage{tikz}
\usetikzlibrary{calc}
\usetikzlibrary{decorations.pathmorphing}
\usetikzlibrary{decorations.markings}
\usetikzlibrary{patterns}
\usetikzlibrary{positioning}
\usetikzlibrary{shadows}

\usepackage{todonotes}

\usepackage[bookmarks=false]{hyperref}

\newcommand{\eps}{\varepsilon}

\newcommand{\tp}{\intercal}

\newcommand{\N}{\mathbb{N}}
\newcommand{\Z}{\mathbb{Z}}

\newcommand{\R}{\mathbb{R}}
\newcommand{\Rnn}{\mathbb{R}_{\geq 0}}

\newcommand{\set}[1]{\left\{#1\right\}}
\newcommand{\setc}[2]{\left\{\left.#1\ \right|\ #2\right\}}

\newcommand{\ceil}[1]{\left\lceil#1\right\rceil}

\newcommand{\eg}{e.g.,\xspace}
\newcommand{\etal}{{et\,al.}\xspace}

\newcommand{\ie}{i.e.,\xspace}
\newcommand{\st}{\text{s.t.}\xspace}

\newcommand{\DeclareMathProblem}[2]{\newcommand{#1}{\ensuremath{\mathsf{#2}}\xspace}}
\DeclareMathProblem{\EC}{EC}
\DeclareMathProblem{\LP}{LP}
\DeclareMathProblem{\IP}{IP}
\DeclareMathProblem{\OPT}{OPT}
\DeclareMathProblem{\SDP}{SDP}
\DeclareMathProblem{\TAP}{TAP}
\DeclareMathProblem{\WTAP}{WTAP}

\newcommand{\DeclareCC}[2]{\newcommand{#1}{\ensuremath{\mathsf{#2}}\xspace}}
\DeclareCC{\APX}{APX}
\DeclareCC{\FPTAS}{FPTAS}
\DeclareCC{\NP}{NP}
\DeclareCC{\PT}{P}
\DeclareCC{\PTAS}{PTAS}

\newenvironment{lp}[2]{
 \begin{array}{ll@{\ \ }c@{\ \ }l@{\quad}l}
  \displaystyle \textrm{#1}  & \displaystyle #2, &      &     & \\[5mm]
  \displaystyle \textrm{\st}
}{   
 \end{array}
}

\newcommand{\lpconstraint}[4]{& \displaystyle #1 & \displaystyle #2 & \displaystyle #3 & \displaystyle \text{for all } #4,\\[2mm]}
\newcommand{\finallpconstraint}[4]{& \displaystyle #1 & \displaystyle #2 & \displaystyle #3 & \displaystyle \text{for all } #4.\\}
\DeclareMathOperator{\supp}{supp}

\newtheorem{theorem}{Theorem}

\newtheorem{claim}[theorem]{Claim}

\newtheorem{lemma}[theorem]{Lemma}
\newtheorem{observation}[theorem]{Observation}

\newcommand{\introduceproblem}[3]{
 ~\\
 \framebox{
  \parbox{0.969\textwidth}{
   \vspace{1mm}
   {\large\textsf{#1}}\\[2mm]
   \begin{tabular}{lp{0.8\textwidth}}
   \emph{Input}:  & #2\tabularnewline[1mm]
   \emph{Output}: & #3
   \end{tabular}
  }
  \\
 }
\vspace{4mm}
}

\DeclareMathOperator{\cov}{cov}
\DeclareMathOperator{\lca}{lca}

\DeclareMathOperator{\leaves}{Leaves}
\DeclareMathOperator{\outdeg}{outdeg}

\newcommand{\zhcgcuts}{$\set{0,\tfrac{1}{2}}$-Chv\'atal-Gomory cuts\xspace}
\newcommand{\xcr}{\ensuremath{x^{\mathsf{cr}}}\xspace}
\newcommand{\xin}{\ensuremath{x^{\mathsf{in}}}\xspace}

\newcommand{\Lup}{\ensuremath{L^{\mathsf{up}}}\xspace}
\newcommand{\Lcr}{\ensuremath{L^{\mathsf{cr}}}\xspace}
\newcommand{\Lin}{\ensuremath{L^{\mathsf{in}}}\xspace}
\newcommand{\oddcutlp}{odd-cut \ensuremath{\mathsf{LP}}\xspace}

\newcommand{\oflp}{odd-cut bundle \ensuremath{\mathsf{LP}}\xspace}
\newcommand{\oflps}{odd-cut bundle \ensuremath{\mathsf{LP}} \ensuremath{(\LP^{odd}_\gamma)}\xspace}

\tikzstyle{node}=[draw,circle,fill=black,scale=0.5]
\tikzstyle{root}=[node, fill=white]
\tikzstyle{edge}=[very thick]
\tikzstyle{cut}=[blue!80]
\tikzstyle{link}=[-, very thick, dashed, blue!80]
\tikzstyle{inlink}=[-, very thick, dashed, blue!80]
\tikzstyle{crosslink}=[-, very thick, dotted, green!75!black!60]
\tikzstyle{uplink}=[-, very thick, dash dot, orange!80]
\tikzstyle{marked}=[ultra thick, red]

\title{A $\frac{3}{2}$-Approximation Algorithm for Tree Augmentation via Chv\'atal-Gomory Cuts}
\author{Samuel Fiorini~\thanks{Universit\'e libre de Bruxelles, Brussels, Belgium. eMail: \texttt{sfiorini@ulb.ac.be}} \and Martin Gro\ss{} \and Jochen K\"onemann \and Laura Sanit\`a~\thanks{University of Waterloo, Waterloo, ON, Canada. eMail: \texttt{\{mgrob,jochen,laura.sanita\}@uwaterloo.ca}}}
\date{}

\begin{document}

\maketitle

\begin{abstract}
The weighted tree augmentation problem (\WTAP) is a fundamental network design problem. We are given an undirected tree $G = (V,E)$, an additional set of edges $L$ called \emph{links} and a cost vector $c \in \R^L_{\geq 1}$. The goal is to choose a minimum cost subset $S \subseteq L$ such that $G = (V, E \cup S)$ is $2$-edge-connected. In the unweighted case, that is, when we have $c_\ell = 1$ for all $\ell \in L$, the problem is called the tree augmentation problem (\TAP).
 
Both problems are known to be \APX-hard, and the best known approximation factors are $2$ for \WTAP by (Frederickson and J\'aJ\'a, '81) and $\tfrac{3}{2}$ for \TAP due to (Kortsarz and Nutov, TALG '16). In the case where all link costs are bounded by a constant $M$, (Adjiashvili, SODA~'17) recently gave a $\approx 1.96418+\eps$-approximation algorithm for \WTAP under this assumption. This is the first approximation with a better guarantee than $2$ that does not require restrictions on the structure of the tree or the links. 

In this paper, we improve Adjiashvili's approximation to a $\tfrac{3}{2}+\eps$-approximation for \WTAP under the bounded cost assumption. We achieve this by introducing a strong \LP that combines \zhcgcuts for the standard \LP for the problem with bundle constraints from Adjiashvili. We show that our \LP can be solved efficiently and that it is exact for some instances that arise at the core of Adjiashvili's approach. This results in the improved guarantee of $\tfrac{3}{2}+\eps$. For \TAP, this is the best known \LP-based result, and matches the bound of $\tfrac{3}{2}+\eps$ achieved by the best \SDP-based algorithm due to (Cheriyan and Gao, arXiv '15). 
\end{abstract}
\section{Introduction}

The tree augmentation problem (weighted or unweighted) is 
a fundamental and intensively studied problem in the area 
of network design, see for example the surveys by Khuller~%
\cite{Khuller97} and Kortsarz and Nutov~\cite{KortsarzNutov10}. 
While already in the
unweighted case the problem is known to be \APX-hard, the best
algorithms for \WTAP and \TAP achieve approximation factors of 
$2$ and $\tfrac{3}{2}$ respectively. 
One of the main open questions about these problems is to improve 
the quality of approximation algorithms.

Adjiashvili~\cite{Adjiashvili17} recently managed to push the
approximation guarantee for \WTAP below $2$ in case the link 
costs are bounded by a constant $M$, which was the first improvement 
in over 35 years that did not restrict the structure of the tree 
or the links. His algorithm is based on an \LP that strengthens the
standard \LP for the problem. Letting $\cov(e)$ denote the set 
of links connecting distinct connected components of $G - e$, for 
each tree edge $e \in E$, the standard \LP for \WTAP is
\begin{alignat}{4}
 & \text{min}    &       & \sum_{\ell \in L}           c_\ell &x_\ell &       & \label{eq:cut1}\\
 & \text{s.\,t.} & \quad & \sum_{\mathclap{\ell \in \cov(e)}} &x_\ell &\geq 1 & \qquad &\text{for all } e\in E,\label{eq:cut2}\\
 &               &       &                                    &x_\ell &\geq 0 & \qquad &\text{for all } \ell \in L.\label{eq:cut3}
\end{alignat}
This \LP is known as the \emph{cut \LP}. 

\paragraph*{Our results.}

We add to the cut \LP all its \zhcgcuts, thus obtaining a new \LP 
for \WTAP that we call the \oddcutlp. The \oddcutlp is key to our
approach. It has three extremely useful properties:

\begin{itemize}
\item One can solve the \oddcutlp efficiently, even though 
separating \zhcgcuts is \NP-hard in general~%
\cite{CapraraFischetti96}.  

\item The \oddcutlp is compatible with the decomposition approach of~\cite{Adjiashvili17} to split the given 
instance and \LP solution into well-structured independent
instances together with their own local \LP solutions.

\item The \oddcutlp is exact if for a certain choice of 
root $r \in V[G]$, every link $\ell$ connects either two 
different connected components of $G - r$ (in which case 
$\ell$ is called a cross-link) or some node of $G$ to one 
of its ancestors (in which case $\ell$ is called an up-link).
\end{itemize}

We prove the last property by establishing that the
constraint matrix of the cut \LP is an integral binet matrix. 
These matrices were introduced by Appa and Kotnyek~%
\cite{AppaKotnyek04} as a generalization of network matrices.
Relying on earlier work by Edmonds and Johnson~\cite{EdmondsJohnson73},
Appa~\etal~\cite{AppaEtAl07} proved that the integer hull of 
polyhedra of the form~$\setc{x}{Ax \geq b,\ x \geq 0}$ can be 
described by \zhcgcuts whenever $A$ is an 
integral binet matrix and $b$ is an integer vector. This 
results in the following theorem.

\begin{restatable}{theorem}{exact}
\label{theorem:exact}
 The odd-cut \LP is integral for \WTAP instances that contain only cross- and up-links.
\end{restatable}

Although the \oddcutlp alone might be sufficient to obtain a 
$\tfrac{3}{2}$-approximation for \WTAP (or maybe even better approximation), we combine the odd-cut \LP 
with the bundle constraints from~\cite{Adjiashvili17}, resulting in 
the \oflp. This last \LP is the one that we use in our algorithm.

We follow the decomposition approach of~\cite{Adjiashvili17}. 
After splitting the given instance and its optimum \LP solution 
into independent rooted instances and corresponding \LP solutions
(respectively), at an extra cost of $\varepsilon \OPT$, 
Adjiashvili applies to each one of the local instances 
two distinct procedures producing feasible solutions, whose 
cost is bounded in terms of the local \LP solution.

One of the two procedures of~\cite{Adjiashvili17} produces an 
integer solution of cost at most $c^\intercal x^\mathrm{in} + 
2 c^\intercal x^\mathrm{cr} + \varepsilon \OPT$, where 
$c^\intercal x^\mathrm{in}$ is the local \LP cost on in-links
(defined as all the links that are not cross-links) and 
$c^\intercal x^\mathrm{cr}$ is the local \LP cost on cross-links. 
This is the part of the analysis where bundle constraints 
are used. We keep this procedure as is in our algorithm.

Using Theorem~\ref{theorem:exact}, we improve the other 
procedure of~\cite{Adjiashvili17} to obtain an integer 
solution of cost at most $2 c^\intercal x^\mathrm{in} + 
c^\intercal x^\mathrm{cr}$. This gives a significant 
improvement in the approximation factor since combining 
both procedures, we can construct an integer solution 
in each of the local instances of cost at most
\begin{align*}
&\ \min\left\{c^\intercal x^\mathrm{in} + 
2 c^\intercal x^\mathrm{cr} + \delta,
2 c^\intercal x^\mathrm{in} + 
c^\intercal x^\mathrm{cr}\right\}\\
\leq &\
\frac12 \left(c^\intercal x^\mathrm{in} + 
2 c^\intercal x^\mathrm{cr} + \delta\right)
+ \frac12 \left(2 c^\intercal x^\mathrm{in} + 
c^\intercal x^\mathrm{cr}\right)\\
\leq&\ \frac32 \left(c^\intercal x^\mathrm{in} + 
c^\intercal x^\mathrm{cr}\right) + \delta,
\end{align*}
where $\delta$ is a small quantity whose sum across the
local instances is at most $\varepsilon \OPT$. This yields 
our main result.

\begin{restatable}{theorem}{OnePointHalf}
 \label{thm:onepointhalf}
 For every fixed $\eps > 0$ and $M \in \R_{\geq 1}$, 
 there exists an \LP-based polynomial time $\tfrac{3}{2}+\eps$-approximation algorithm for \WTAP with link costs in $[1,M]$.
\end{restatable}

This result is the best known for \WTAP with link costs bounded 
by a constant, and a significant improvement over the previously 
known $\approx 1.96418+\eps$-approximation~\cite{Adjiashvili17}. 
For \TAP, it is the best \LP-based result and matches the results 
of the \SDP-based algorithm of~\cite{CheriyanGao15}, while only 
being the $\eps$-term worse than the best overall algorithm~%
\cite{KortsarzNutov16}. Finally, we point out that while the
approximation factors are improving, the proofs are actually 
getting simpler, which we see as another indication of the 
power behind our approach.
 
\paragraph*{Related work.} 
Frederickson and J\'aJ\'a~\cite{FredericksonJaja81} showed that \WTAP is \NP-hard even if the tree has constant diameter and the link costs are either 1 or 2. For \TAP, Cheriyan \etal~\cite{CheriyanEtAl99} proved that the problem is \NP-hard even if the links form a cycle on the leaves of the tree. Kortsarz, Krauthgamer and Lee~\cite{KortsarzKrauthgamerLee04} then showed that even \TAP is \APX hard, meaning that these problems have no \PTAS, unless $\PT = \NP$.

For \WTAP, there are several 2-approximations known, the first given by Frederickson and J\'aj\'a~\cite{FredericksonJaja81} and then simplified by Khuller and Thurimella~\cite{KhullerThurimella93}. The primal-dual approach of Goemans \etal~\cite{GoemansEtAl94} and the iterative rounding algorithm of Jain~\cite{Jain01} also give a 2-approximation for this problem. Until recently, the factor of 2 was best known for general trees. Adjiashvili~\cite{Adjiashvili17} then managed to give a $\delta+\eps$-approximation for the case where all link costs are bound by a constant, for any small $\eps > 0$ and $\delta = \tfrac{8(23+3\sqrt{5})}{121} \approx 1.96418$.

For \TAP, the best known approximation is a $\tfrac{3}{2}$-approximation algorithm by Kortsarz and Nutov~\cite{KortsarzNutov16}, which is purely combinatorial. There is also a $\tfrac{3}{2}+\eps$-approximation for any $\eps > 0$, given by Cheriyan and Gao~\cite{CheriyanGao15}, which is also combinatorial but whose analysis is based on an \SDP relaxation, showing an integrality gap of $\tfrac{3}{2}+\eps$ for the \SDP. As far as \LP-based algorithms are concerned, the state of the art was a $\tfrac{7}{4}$-approximation due to Kortsarz and Nutov~\cite{KortsarzNutov16B} until Adjiashvili~\cite{Adjiashvili17} gave a $\tfrac{5}{3}+\eps$-approximation based on the bundle \LP, for any small $\eps > 0$.

For special classes of trees, there are better approximation guarantees known. For example, Cohen and Nutov~\cite{CohenNutov13} described a $1+\ln 2 \approx 1.6931$-approximation for \WTAP for trees with constant radius.
For the special case of \TAP where every link ends in two leaves of the tree, Maduel and Nutov~\cite{MaduelNutov10} gave a $\tfrac{17}{12}$-approximation. In case that the tree has radius 3 or 2, they could strengthen the bound to $\tfrac{11}{8}$ and $\tfrac{4}{3}$, respectively.

In terms of lower bounds on integrality gaps, Cheriyan \etal~\cite{CheriyanEtAl99} conjectured that the cut \LP has an integrality gap of $\tfrac{4}{3}$ for \WTAP. However, this conjecture has been refuted by Cheriyan \etal~\cite{CheriyanEtAl08}, who proved that the integrality gap of the cut \LP is at least $\tfrac{3}{2}$, even for \TAP. 
\section{Preliminaries\label{sec:preliminaries}}
We begin by restating the definition of the weighted tree augmentation problem (\WTAP). Recall that a graph is 2-edge-connected if and only if there are at least two edge-disjoint paths between all pairs of nodes. 

\introduceproblem{(Weighted) Tree Augmentation Problem (\WTAP)}
{An undirected tree $G = (V,E)$, an additional set of edges $L$ on $V$ called \emph{links}, and a cost vector $c \in \R^L_{\geq 0}$.}
{A minimum cost set of links $S$ such that $G = (V, E \cup S)$ is 2-edge-connected.}

The tree augmentation problem (\TAP) is then the special case where $c_\ell = 1$ for all $\ell \in L$. Notice that we do not demand that $E$ and $L$ are disjoint; consequently, we allow parallel edges in the union of $E$ and $S$. 
We assume without loss of generality that $c_\ell > 0$ for all $\ell \in L$, as we can always pick links with zero cost. Let $c_{min}$ and $c_{max}$ be the minimum and maximum link costs, then we scale the link costs by $1/c_{min}$, resulting in $c_\ell \in [1,c_{max}/c_{min}]$ for all $\ell \in L$. For \WTAP, we study the case that link costs are bounded from above by a constant $M \in \R_{\geq 1}$, \ie $c_\ell \leq M$ for all $\ell \in L$. For a graph $H$, we denote the set of nodes and edges by $V[H]$ and $E[H]$, respectively. For a set of nodes $V' \subseteq V[G]$, we refer to the set of edges between nodes in $V'$ by $E[V'] := \setc{e =\set{u,v}\in E[G]}{u,v \in V'}$. For a vector $x \in \R^N_{\geq 0}$, let $\supp(x) := \setc{i \in N}{x_i > 0}$ denote the support of $X$. For a subset $N' \subseteq N$, define $x(N') := \sum_{i \in N'} x_i$. Furthermore, we set $[k] := \set{1,\dots,k}$.

\paragraph*{Links.}
We write $e = \set{u,v}$ for an edge $e \in E$ connecting nodes $u$ and $v$, and we write $\ell = uv$ for a link connecting nodes $u$ and $v$. Since $G$ is a tree, there is a unique path between two nodes $u,v \in V[G]$. For an $\ell = uv \in L$, we refer to this path by $P_\ell^G$ and call all edges $e \in P_\ell^G$ \emph{covered} by $\ell$, since $P_\ell^G$ together with $\ell$ is 2-edge-connected. If $G$ is clear due to the context, we omit it. For a set $F \subseteq E[G]$, we define $\cov(F)$ to be the set of links covering at least one edge of $F$. For brevity, we write $\cov(e)$ instead of $\cov(\set{e})$. For a set of links $L' \subseteq L$ and a set of edges $E' \subseteq E$, we say that $L'$ covers $E'$ if every edge $e \in E'$ is covered by at least one link $\ell \in L'$. The concept of covering is highly relevant due to the following observation.

\begin{observation}
 \label{obs:feasibility}
 Given an instance $(G,L,c)$ of \WTAP, a set of links $S \subseteq L$ is a feasible solution if and only if every edge is covered by a link $\ell \in S$, \ie $\bigcup_{\ell \in S} P_\ell = E[G]$.
\end{observation}


\paragraph*{Up-links, in-links, cross-links.}
For the following classification of links, we assume that our tree is rooted at an arbitrary node $w$. For a link $\ell = uv \in L$, let $\lca(uv) \in V[G]$ denote the \emph{least common ancestor} of $u$ and $v$ in $G$. If $\lca(uv) \not\in \set{u,v}$ and $\lca(uv) = w$, we call $\ell$ an \emph{cross-link}. Otherwise, we call $\ell$ an \emph{in-link}. We call an in-link $\ell$ an \emph{up-link}, if $\lca(uv) \in \set{u,v}$. Notice that links $\ell$ with $w \in \ell$ are up-links by this definition. We denote the set of up-links, cross-links and in-links by $\Lup$, $\Lcr$ and $\Lin$, respectively. Figure~\ref{figure:contraction} illustrates these types of links.

In order to analyze the approximation guarantee later on, we will need to study the cost carried by different types of links separately. Let $x \in \Rnn^L$ be a fractional solution to a \WTAP instance. Then we split $x$ into the parts belonging to cross-links $\xcr$ and in-links $\xin$ as follows: 
\[
 \xcr_\ell = \begin{cases} x_\ell & \ell \in \Lcr\\ 0 & \text{else}\end{cases}, \quad \xin_\ell = \begin{cases} x_\ell & \ell \in \Lin\\ 0 & \text{else} \end{cases} \quad \text{for all } \ell \in L.
\]

\begin{figure}[htbp]
 \centering
 \begin{tikzpicture}
	\begin{scope}[xshift=0cm]
   \node[root] (r) at (0,0) {};
	 \node[node] (v1) at (90:1) {};
	 \node[node] (v2) at (210:1) {};
	 \node[node] (v3) at (330:1) {};
	 \node[node] (v11) at ($(v1)+(-1,0)$) {};
	 \node[node] (v12) at ($(v1)+( 1,0)$) {};
	 \node[node] (v13) at ($(v1)+(-0.6,1)$) {};
	 \node[node] (v14) at ($(v1)+( 0.6,1)$) {};	
	 \node[node] (v21) at ($(v2)+(-1,0)$) {};
	 \node[node] (v31) at ($(v3)+( 1,0)$) {};
	 \node[node] (v32) at ($(v3)+( 0,-1)$) {};
	 \draw[edge] (r) edge (v1) edge (v2) edge (v3);
	 \draw[edge] (v1) edge (v12) edge (v14);
	 \draw[edge] (v2);
	 \draw[edge] (v3) edge (v32);	
	 \draw[edge] (v1) edge (v11) edge (v13);	
	 \draw[edge] (v2) edge (v21);
	 \draw[edge] (v3) edge (v31);
	 \draw[crosslink] (v21) edge (v11);
	 \draw[uplink] (r) edge (v14);
	 \draw[uplink] (r) edge (v21);
	 \draw[inlink] (v31) edge (v32);
	\end{scope}
	\begin{scope}[xshift=4cm, every loop/.style={}]
   \node[root] (r) at (0,0) {};
	 \coordinate (v1) at (90:1);
	 \coordinate (v2) at (210:1);
	 \node[node] (v3) at (330:1) {};
	 \node[node] (v12) at ($(v1)+( 1,0)$) {};
	 \node[node] (v13) at ($(v1)+(-0.6,1)$) {};
	 \node[node] (v14) at ($(v1)+( 0.6,1)$) {};	
	 \node[node] (v31) at ($(v3)+( 1,0)$) {};
	 \node[node] (v32) at ($(v3)+( 0,-1)$) {};
	 \draw[edge] (r) edge (v12) edge (v13) edge (v14) edge (v3);
	 \draw[edge] (v3) edge (v31) edge (v32);	
	 \draw[crosslink] (r) edge[in=140,out=200,loop] ();
	 \draw[uplink] (r) edge[bend right=10] (v14);
	 \draw[uplink] (r) edge[in=230,out=290,loop] ();
	 \draw[inlink] (v31) edge (v32);	
	\end{scope}	
	\begin{scope}[xshift=9.5cm]
   \node[node] (r) at (0,0) {};
	 \node[node] (v1) at (90:1) {};
	 \node[node] (v2) at (210:1) {};
	 \node[node] (v3) at (330:1) {};
	 \node[node] (v11) at ($(v1)+(-1,0)$) {};
	 \node[node] (v12) at ($(v1)+( 1,0)$) {};
	 \node[node] (v13) at ($(v1)+(-0.6,1)$) {};
	 \node[node] (v14) at ($(v1)+( 0.6,1)$) {};	
	 \node[node] (v21) at ($(v2)+(-1,0)$) {};
	 \node[node] (v31) at ($(v3)+( 1,0)$) {};
	 \node[node] (v32) at ($(v3)+( 0,-1)$) {};
	 \draw[edge] (r) edge (v2) edge (v3);
	 \draw[edge] (v1) edge (v12) edge (v14);
	 \draw[edge] (v2);
	 \draw[edge] (v3) edge (v32);	
	 \draw[marked] (r) edge (v1);
	 \draw[marked] (v1) edge (v11) edge (v13);	
	 \draw[marked] (v2) edge (v21);
	 \draw[marked] (v3) edge (v31);
	\end{scope}		
 \end{tikzpicture}
 \caption{The picture on the left shows a \TAP instance rooted at the white node with edges drawn solidly and links with dashed and dotted lines. The dotted link is a cross-link, the dashed link is an in-link and the dash-dotted links are up-links (and also in-links). The picture in the middle shows the same instance after the contraction of the cross-link, and the picture on the right shows a 4-bundle (the thick edges).\label{figure:contraction}}
\end{figure}
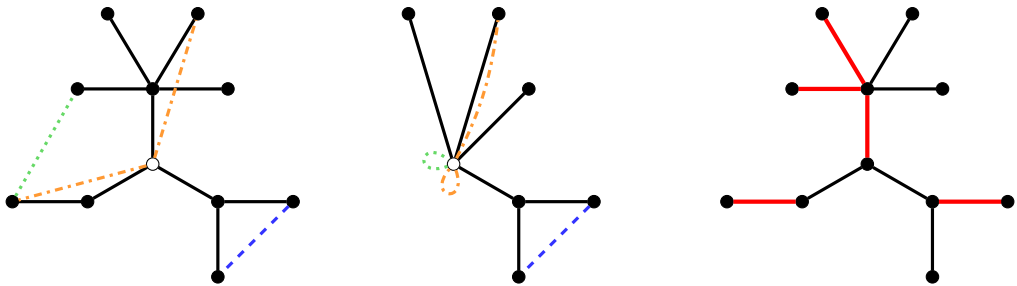

\paragraph*{Contractions.}
By contracting an edge $e = \set{u,v}$ we mean that $u$ and $v$ are replaced by a new node $w$. All edges $e' \in E(G)$ with either $u \in e$ or $v \in e$ are modified by replacing $u$ or $v$ with $w$, respectively. The edge $e = \set{u,v}$ is deleted. Links $\ell$ with either $u \in \ell$ or $v \in \ell$ are modified in the same way, but links $\ell = uv$ become self-loops of $w$ instead of being deleted. 

Contracting a set of edges $F$ is defined by contracting all edges $e \in F$, in any order. Contracting a link $\ell \in L$ refers to contracting all edges in $P_\ell$. Figure~\ref{figure:contraction} illustrates this.

\paragraph*{Bundles and the bundle \LP.}
The concept of bundles and the bundle \LP is due to~\cite{Adjiashvili17}. A \emph{$\gamma$-bundle} in a graph $G$ is the union of $\gamma$ paths in $G$ for an integer $\gamma \in \N$. We denote the set of all $\gamma$-bundles in $G$ by $\mathcal{B}_\gamma$. Notice that the paths of a $\gamma$-bundle do not have to be distinct, so $\mathcal{B}_\gamma \subseteq \mathcal{B}_{\gamma+1}$. Also, if $G$ is a tree, there are at most $\binom{|V(G)|}{2}$ distinct paths in $G$. Therefore, we have $|\mathcal{B}_\gamma| \leq \binom{|V(G)|}{2}^\gamma$ which is polynomial in the input size if $\gamma$ is a constant. Figure~\ref{figure:contraction} shows a $4$-bundle.

The \emph{bundle} \LP $(\LP_\gamma)$ adds the following constraints for every bundle in $\mathcal{B}_\gamma$ for a constant $\gamma \in \N$ to the cut \LP \eqref{eq:cut1} -- \eqref{eq:cut3}:
\begin{equation}
 \label{eq:bundle-LP}
 \sum_{\ell \in \cov(B)} c_\ell x_\ell \geq \OPT(B) \quad \text{ for all } B \in \mathcal{B}_\gamma .
\end{equation}
Above, $\OPT(B)$ is the minimum cost of any set of links in $L$ that covers $B$. These constraints are clearly valid for any integral solution, and we can also compute $\OPT(B)$ efficiently using the following lemma.

\begin{lemma}[Lemma A.1 in~\cite{Adjiashvili17}]
 \label{lemma:constantleaves}
 Let $(G,L,c)$ be a \WTAP instance, with $k$ being the number of leaves of $G$. An optimal solution to $(G,L,c)$ can be computed in time ${n^k}^{O(1)}$.
\end{lemma}

For a $\gamma$-bundle $B$ in $G$, we know that $B$ is a forest with at most $2\gamma$ leaves. In order to apply this lemma, we contract all edges in $E[G] \setminus B$ to obtain an equivalent tree with at most $2\gamma$ leaves. Then we can apply Lemma~\ref{lemma:constantleaves} to obtain an optimal solution. 

We will later consider the case where we have the set of $\gamma$-bundles $\mathcal{B}_\gamma$ of a graph $G$, and then contract edges in $E[G]$ to obtain a subgraph $G'$. In this case, we are interested in the set of edges in $G'$ that are a $\gamma$-bundle in $G$. We refer to the set of edges in $G'$ that are $\gamma$-bundles in $G$ by $\mathcal{B}_\gamma^{G'}$. Notice that due to the contractions, it is possible that a $\gamma$-bundle in $G'$ is not a $\gamma$-bundle in $G$.
\section{A Stronger \LP}

In this section, we introduce the \oflps which strengthens
the bundle \LP used in~\cite{Adjiashvili17}. We use our 
\LP in the algorithmic framework described in 
Section~\ref{sec:decomposition} to obtain better 
approximation guarantees. The additional constraints that 
we need are Chv\'atal-Gomory cuts obtained from the cut \LP.

\subsection{Deriving and Separating the Odd-Cut Constraints}

\paragraph*{Derivation.}
Consider an \IP of the form $\min \{c^\intercal x \mid Ax \geq b,\ 
x \in \Z^n\}$, with $A \in \R^{m \times n}$ and $b \in \R^{m}$. 
A \emph{Chv\'atal-Gomory cut} \cite{Chvatal73,Gomory58} is a
constraint of the form $\lambda^\intercal A x \geq \lceil 
\lambda^\intercal b \rceil$, 
where the vector of multipliers $\lambda \in \R^m_{\geq 0}$ is 
chosen in such a way that $\lambda^\intercal A \in \Z^n$. Clearly, 
any such constraint is valid for the integer solutions of the \IP. 
It is well known that the cuts obtained for $\lambda \in [0,1)^m$
imply the cuts obtained for larger multipliers, hence one can 
assume $\lambda \in [0,1)^m$ without loss of generality. (A proof of
this fact and more background on the Chv\'atal-Gomory cuts can be
found, for instance, in~\cite{CCZ2014}.) In case we restrict further 
$\lambda$ to be in $\{0,\frac{1}{2}\}^m$, we obtain a 
\emph{$\{0,\frac{1}{2}\}$-Chv\'atal-Gomory cut}~%
\cite{CapraraFischetti96}. These cuts are a proper specialization 
of the Chv\'atal-Gomory cuts, and are precisely the cuts that we 
use here.

A $\{0,\frac{1}{2}\}$-Chv\'atal-Gomory cut for the cut \LP is any
constraint of the form
\begin{equation}
\label{eq:0-1/2-cut} \sum_{e \in E[G]} \lambda_e x(\cov(e)) + \sum_{\ell \in L} \mu_\ell x_\ell \geq \Big\lceil\sum_{e \in E[G]}\lambda_e\Big\rceil
\end{equation}
where $\lambda \in \{0,\frac{1}{2}\}^{E[G]}$ and $\mu \in \{0,
\frac{1}{2}\}^{L}$ are such that the coefficients in the left-hand
side are all integral. Notice that for any fixed $\lambda$, there is a unique $\mu$ that achieves this.

Let $K := \supp(\lambda) = \{e \in E[G] \mid \lambda = \frac{1}{2}\}$. Since $G$ is a tree, there exists a (not necessarily connected) set $S \subseteq V[G]$ such that $K = \delta_G(S)$. Notice that the right-hand side of \eqref{eq:0-1/2-cut} is $\lceil |\delta_G(S)|/2 \rceil$, and that the cut is redundant whenever $|\delta_G(S)|$ is even.

Let $\pi(S)$ denote the \emph{multiset} of links $\ell$ such that $P_\ell$ intersects $\delta_G(S)$, in which the multiplicity of $\ell$ is defined as $\lceil \frac{1}{2} |P_\ell \cap \delta_G(S)| \rceil$ (see Figure~\ref{fig:pi_of_S}). Now, assuming that $|\delta_G(S)|$ is odd, we can rewrite the constraint~\eqref{eq:0-1/2-cut} as 
\begin{equation}
\label{eq:0-1/2-cut_pi}
x(\pi(S))\geqslant \frac{|\delta_G(S)|+1}2\,.
\end{equation}

\begin{figure}[htb]
 \begin{center}
  \begin{tikzpicture}
	 \draw[blue!80,thick,fill=blue!20] (1.5,4.3) circle (1.2cm);
	 \draw[blue!80,thick,fill=blue!20] (4.8,4.1) circle (1.2cm);
	 \draw[blue!80,thick,fill=blue!20] (3.2,1.7) circle (1.1cm);
	 \node[blue!80,scale=1.0] (text) at (1.5,5.0) {$S_1$};
	 \node[blue!80,scale=1.0] (text) at (4.8,4.8) {$S_2$};
	 \node[blue!80,scale=1.0] (text) at (3.2,2.4) {$S_3$};
	 \node[node] (a) at (0.0,5.0) {};
	 \node[node] (b) at (0.3,3.0) {};
	 \node[node] (c) at (0.6,4.5) {};
	 \node[node] (d) at (0.9,4.9) {};
	 \node[node] (e) at (1.2,4) {};
	 \node[node] (f) at (1.4,3.4) {};
	 \node[node] (g) at (1.5,2.6) {};
	 \node[node] (h) at (1.6,1.4) {};
	 \node[node] (i) at (1.8,4.3) {};
	 \node[node] (j) at (1.9,4.8) {};
	 \node[node] (k) at (2.1,3.0) {};
	 \node[node] (l) at (2.2,3.6) {};	 
	 \node[node] (m) at (2.4,0.6) {};
	 \node[node] (n) at (2.5,4.2) {};
	 \node[node] (o) at (2.5,1.4) {};
	 \node[node] (p) at (2.8,3.2) {};
	 \node[node] (q) at (3.2,1.1) {};
	 \node[node] (r) at (3.2,2.0) {};
	 \node[node] (s) at (3.3,4.4) {};
	 \node[node] (t) at (3.6,5.0) {};
	 \node[node] (u) at (3.8,1.4) {};
	 \node[node] (v) at (4.0,2.2) {};
	 \node[node] (w) at (4.2,2.8) {};
	 \node[node] (x) at (4.2,3.4) {};
	 \node[node] (y) at (4.0,4.1) {};
	 \node[node] (z) at (4.8,3.1) {};
	 \node[node] (aa) at (5.1,4.2) {};
	 \node[node] (ab) at (5.3,3.6) {};
	 \node[node] (ac) at (5.4,4.5) {};
	 \node[node] (ad) at (5.6,3.8) {};
	 \node[node] (ae) at (6.0,3.0) {};
	 \node[node] (af) at (6.3,3.4) {};
	 \node[node] (ag) at (6.0,5.1) {};
	 \draw[edge,cut] (a) edge (c);
	 \draw[edge,cut] (b) edge (f);
	 \draw[edge] (c) edge (d) edge (e);
	 \draw[edge] (e) edge[bend right, link, green!60!black!100] node[auto,swap,scale=0.8,pos=0.75] {1}  (b) edge (f) edge (i);
	 \draw[edge] (i) edge (j) edge (l) edge (n);
	 \draw[edge,cut] (l) edge (k) edge (p);
	 \draw[edge] (g) edge (k) edge[bend right, link, red] node[auto,swap,scale=0.8] {3}  (h);
	 \draw[edge] (s) edge[cut] (n) edge (t) edge[cut] (y) edge[bend right, link, red] node[auto,swap,scale=0.8] {1}  (w);
	 \draw[edge] (x) edge[cut] (w) edge (y) edge (z);
	 \draw[edge] (v) edge[cut] (w) edge (r);
	 \draw[edge] (o) edge[cut] (h) edge[cut] (m) edge (r);
	 \draw[edge] (aa) edge (y) edge (ab) edge (ac) edge (ad);
	 \draw[edge] (r) edge (q) edge (u);
	 \draw[edge] (af) edge[bend left=50, link, green!60!black!100] node[auto,scale=0.8] {2}  (q);
	 \draw[edge] (ad) edge[cut] (ae) edge[cut] (af) edge[bend left,link, red] node[auto,scale=0.8] {1} (u);
	 \draw[edge] (ag) edge[bend right, link, red] node[auto,swap,scale=0.8] {2}  (a) edge[cut] (ac);
	 
	\end{tikzpicture}  
 \end{center}
\caption{A set $S=S_1\cup S_2\cup S_3$ of vertices and the
  corresponding connected components it induces on tree $G$. Edges
  $e \in \delta_G(S)$ are blue. Links $\ell \in \delta_L(S)$
  are in green. The other links are in red. For each link $\ell$,
  the figure gives its multiplicity in $\pi(S)$.
}
\label{fig:pi_of_S}
\end{figure}
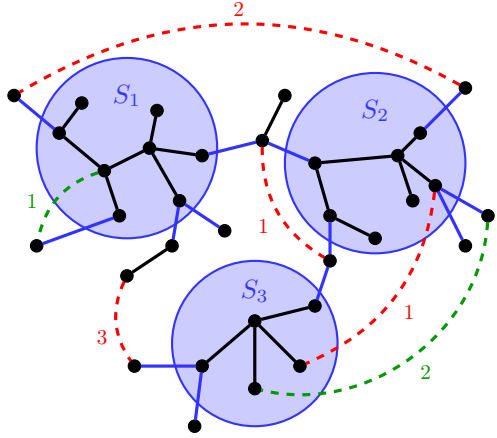

Let $\mathcal{S}$ denote the collection of all sets $S$ such that $|\delta_G(S)|$ is odd, and let the \emph{\oddcutlp} be the \LP resulting from the cut \LP after adding the \emph{odd-cut constraint} \eqref{eq:0-1/2-cut_pi} for each set $S \in \mathcal{S}$:

\[
 \begin{lp}{min}{\sum_{\ell \in L} c_\ell x_\ell}
  \lpconstraint{x(\pi(S))}{\geq}{\frac{|\delta_G(S)|+1}{2}}{S \in \mathcal{S}}
  \finallpconstraint{x_\ell}{\geq}{0}{\ell\in L}
 \end{lp}
\]

We remark that the first set of constraints includes in 
particular the covering constraint $x(\cov(e)) \geq 1$ for 
each tree edge $e \in E[G]$. Indeed, if $S$ denotes any of 
the two connected components arising when $e$ is deleted from 
$G$, we have $\delta_G(S) = \{e\}$ and thus $S \in \mathcal{S}$.
The corresponding odd-cut constraint is the covering 
constraint.

\paragraph*{Separation.}
A nice fact that is key to our approach is that separation 
of the $\{0,\frac{1}{2}\}$-Chv\'atal-Gomory cuts arising 
from an \IP $\min \{c^\intercal x \mid Ax \geq b,\ x \in \Z^n\}$ 
(here we take $A \in \Z^{m \times n}$ and $b \in \Z^m$) can be 
done in polynomial time whenever the matroid represented over 
$GF(2)$ by $(I\,\bar{A})$ is graphic (or co-graphic), where 
$\bar{A} := A \pmod{2}$ is the parity matrix of $A$. This was
proved by Caprara and Fischetti~\cite{CapraraFischetti96}, and
implies directly that the constraints of the \oddcutlp can be
separated in polynomial time.

A more direct way to prove this is to rewrite inequality 
\eqref{eq:0-1/2-cut_pi} as the \emph{$T$-cut constraint}%
\footnote{Taking $T$ to be the set of odd-degree nodes of
the tree $G$.} $x(\delta_L(S)) + y(\delta_G(S)) \geq 1$ 
after introducing the slack variables $y_e := x(\cov(e)) - 1 \geq 0$ for $e \in E[G]$, 
and then solve the separation problem with the algorithm of 
Padberg and Rao~\cite{PadbergRao82}.

To see this, let $S \in \mathcal{S}$ and let $H$ be the graph obtained from adding all links in $L$ to $G$. Assuming that $y_e = x(\cov(e)) - 1$ for all $e \in E[G]$, we can rewrite the odd-cut constraint~\eqref{eq:0-1/2-cut_pi} as:
\begin{eqnarray*}
  x(\pi(S)) & \geq & \frac{|\delta_G(S)| + 1}2\\
  \iff \frac12\sum_{e \in \delta_G(S)}x(\cov(e)) + \frac12
  x(\delta_L(S)) & \geq &  \frac{|\delta_G(S)| + 1}2\\
  \iff \sum_{e \in \delta_G(S)} \underbrace{(x(\cov(e)) - 1)}_{=y_e} + 
  x(\delta_L(S)) & \geq & 1\\
  \iff (x,y)(\delta_H(S)) & \geq & 1\,.
\end{eqnarray*}
Above, the second form is the original expression of the odd-cut constraint as a $\{0,\frac{1}{2}\}$-Chv\'atal Gomory cut, see~\eqref{eq:0-1/2-cut}. 
 
\subsection{Exactness of the Odd-Cut \LP when all the In-links are Up-links}

An integer matrix $M$ is said to be the \emph{incidence matrix of a bidirected graph} if $\sum_i |M_{ij}| \leq 2$ for every fixed column index $j$. A \emph{binet matrix} is any matrix of the form $B = S^{-1} R$ where $M = (S\, R)$ is the incidence matrix of a bidirected graph with full row-rank and $R$ is a basis of $M$. Binet matrices are a generalization of network matrices and were introduced by Appa and Kotnyek~\cite{AppaKotnyek04}.

Appa \etal~\cite{AppaEtAl07} proved the following result, extending results of Edmonds and Johnson~\cite{EdmondsJohnson70,EdmondsJohnson73} for incidence matrices of bidirected graphs.   

\begin{theorem} \label{theorem:binet_zhcgcuts}
For every binet matrix $A \in \Z^{m \times n}$ and every vector $b \in \Z^{m}$, the integer hull of the polyhedron $P := \setc{x \in \R^n}{Ax \geq b,\ x \geq 0}$ is described by its \zhcgcuts.
\end{theorem}

Pick any root $r \in V[G]$ in tree $G$. We are now ready to recall 
and prove Theorem~\ref{theorem:exact}, which provides a key property
of the \oddcutlp that we use in our approach to approximate \WTAP.

\exact*

\begin{proof}
Let $(G,L,c)$ denote the \WTAP instance and let $r$ denote
the chosen root. For a directed edge $e = (u,v)$, let $z = z(e)
\in \R^{V[G] \setminus \{r\}}$ denote the truncated \emph{incidence 
vector} of $e$ (since
the row corresponding to root $r$ is removed), defined as $z_a := 1$ if $a = u$, $z_a := -1$ 
if $a = v$ and $z_a = 0$ otherwise, for $a \in V[G] \setminus 
\{r\}$.

Direct all the edges of $G$ away from the root and let $R 
\in \R^{(V[G] \setminus \{r\}) \times E[G]}$ denote the 
(truncated) \emph{incidence matrix} of the resulting directed 
graph $\overrightarrow{G}$. Every directed edge $e = (u,v)$ of 
$\overrightarrow{G}$ has a corresponding column $z(e)$ in 
$R$. 

Now, we define a matrix $S \in \R^{(V[G] \setminus \{r\}) 
\times L}$ encoding the links of \WTAP instance. Each link
$\ell = uv\in L$ has a corresponding column $S_{\star \ell}$ 
in $S$. If $\ell = uv$ is an up-link with $u = \lca(uv)$, we let 
$S_{\star \ell}$ be the incidence vector $z(u,v)$ of the directed 
edge $(u,v)$. If $\ell = uv$ is a cross-link, then $u, v \neq
r$ and the corresponding column $S_{\star \ell}$ has $S_{a \ell}
:= 1$ if $a \in \{u,v\}$ and $S_{a \ell} := 0$ otherwise.

Consider the matrix $M = (S\, R)$. Since $\sum_{a \neq r} 
|S_{a\ell}| \leq 2$ for all $\ell \in L$ and $\sum_{a \neq r} 
|R_{ae}| \leq 2$ for all $e \in E[G]$, we see that $M$ is the
incidence matrix of a bidirected graph. Moreover, $M$ has $R$ 
as a basis. Thus $B := R^{-1} S$ is a binet matrix. We claim
that $B \in \R^{(V[G] \setminus \{r\}) \times L} \cong 
\R^{E[G] \times L}$ is the constraint matrix $A$ of the cut \LP 
$\min \{c^\intercal x \mid Ax \geq 1,\ x \geq 0\}$.

In other words, we claim that $A = R^{-1}S$, or equivalently
$R A = S$. This actually is clear: For every link $\ell = uv$, 
the sum of the columns $z(e)$ of $R$ that correspond to the tree 
edges $e = (u,v)$ that are in $P_\ell$ is the corresponding 
column of $S$.

By Theorem~\ref{theorem:binet_zhcgcuts}, the cut \LP becomes integral after one round of $\{0,\frac{1}{2}\}$-Chv\'atal Gomory cuts. That 
is, the \oddcutlp is integral.
\end{proof}

We point out that for these \WTAP instances, there is a 
combinatorial algorithm that finds an optimum solution in
strongly polynomial time. This follows from our proof of 
Theorem~\ref{theorem:exact} and an algorithm of Edmonds
and Johnson~\cite{EdmondsJohnson70,EdmondsJohnson73}, see 
also~\cite{AppaEtAl07}.

\subsection{The Odd-Cut Bundle \LP}

As its name indicates, the \oflps contains all of the constraints of the \oddcutlp and additionally the bundle constraints for a constant $\gamma \in \N$. As before, let $\mathcal{S}$ denote the collection of all sets $S \subseteq V[G]$ such that $|\delta_G(S)|$ is odd, let $\mathcal{B}_\gamma$ be the set of all $\gamma$-bundles, and let $\OPT(B)$ the cost of an integral optimal solution for the \WTAP instance obtained from contracting all edges not in $B$ (with respect to the given costs $c_\ell$). Then, the \oflp is given by:
\[
 \begin{lp}{min}{\sum_{\ell \in L} c_\ell x_\ell}
  \lpconstraint{x(\pi(S))}{\geq}{\frac{|\delta_G(S)|+1}{2}}{S \in \mathcal{S}}
  \lpconstraint{\sum_{\ell \in \cov(B)} c_\ell x_\ell}{\geq}{\OPT(B)}{B \in \mathcal{B}_\gamma}
  \finallpconstraint{x_\ell}{\geq}{0}{\ell\in L}
 \end{lp}
\]
From Lemma~\ref{lemma:constantleaves} and the discussion above, we obtain the following result.

\begin{lemma} \label{lem:solving_main_lp}
For any constant $\gamma \in \N$, the \oflp can be solved in polynomial time.
\end{lemma}


\section{Decomposition of \LP solutions\label{sec:decomposition}}
In this section, we will discuss how we use solutions of the \oflp to solve a \WTAP instance approximately. We use the approach of~\cite{Adjiashvili17}, with two major differences:
\begin{itemize}
 \item We use the \oflps with $\gamma = \ceil{\tfrac{28M}{\eps^2}}$ instead of the bundle \LP as basis for the decomposition and subsequent rounding.
 \item The additional $\set{0,\tfrac{1}{2}}$-Chv\'atal-Gomory cuts in the \oflp compared to the bundle \LP allow us to round instances with only up-links and cross-links without increasing cost. Since in-links can be split into two up-links, this yields a rounding that increases the cost incurred by in-links by a factor of 2 and leaves the cost of cross-links unchanged. This replaces a rounding based on a reduction to edge cover in~\cite{Adjiashvili17} that increases the cost of in-links by a factor of $2\lambda$ and the cost of cross-links by a factor $\tfrac{4}{3}\tfrac{\lambda}{\lambda-1}$ for some $\lambda > 1$.
\end{itemize}

For the decomposition, we will assume that we are given a \emph{shadow-complete} \WTAP instance $(G,L,c)$ and a fractional solution $x$ to the \oflps for $\gamma = \ceil{\tfrac{28M}{\eps^2}}$. An instance $(G,L,c)$ is \emph{shadow-complete}, if for every link $\ell \in L$ all its \emph{shadows} -- \ie links $\ell'$ with $P_{\ell'} \subseteq P_\ell$ -- are also in $L$.
\footnote{We can assume without loss of generality that the instance is shadow-complete -- otherwise, we can add the missing shadows with a cost equal to the cheapest original link of which they are a shadow. Should these links appear in any solution, we can replace them with the original link (which covers even more edges) at no cost.} We assume that the cost of a shadow $\ell'$ of $\ell$ fulfills $c_{\ell'} \leq c_\ell$ -- otherwise, we can set $c_{\ell'} := c_\ell$.
The decomposition begins by computing a solution $L^h$ that covers all \emph{heavily-covered} edges 
\[ E^h := \setc{e \in E}{x(\cov(e)) \geq \tfrac{2}{\eps}} \]
then contracts $L^h$ and obtains a new tree $\overline{G}$ that only contains edges $e$ with $x(\cov(e)) < \tfrac{2}{\eps}$. Then, we decompose $\overline{G}$ and $x$ by splitting repeatedly along \emph{$\alpha$-thin} edges, which are defined as follows.

For an edge $e = \set{u,v}$, let $G^u$ and $G^v$ be the sub-trees of $\overline{G}$ that are obtained by deleting $e$, with $G^u$ containing $u$ and $G^v$ containing $v$, respectively.
An edge $e = \set{u,v}$ is \emph{$\alpha$-thin} with respect to $x$ for an $\alpha \in \R_{\geq 0}$ if the costs of links on both sides of the edge are at least $\alpha$ (we ignore links covering $e$ here):
\begin{equation}\label{eq:alphathin}
 \sum_{\ell \in L, \ell \subseteq V[G^u]} c_\ell x_\ell \geq \alpha, \quad \sum_{\ell \in L, \ell \subseteq V[G^v]} c_\ell x_\ell \geq \alpha.
\end{equation}
If $\overline{G}$ has an $\alpha$-thin edge $e = \set{u,v}$ with respect to $x$, we \emph{split $\overline{G}$ and $x$ along $e$}:

 We obtain $x^u$ and $x^v$ by splitting all links covering $e$ in $x$ into two parts, the part in $G^u$ and the part in $G^v$: 
\[
 x^u_{\ell=pq}\ := \begin{cases} 
                     x_\ell                                      & \text{if } p,q \in V[G^u] \setminus\set{u},\\ 
										 0                                           & \text{if } p \in V[G^v]\text{ or }q \in V[G^v], \\ 
										 x_\ell + \displaystyle\sum\limits_{\substack{\ell' \in \cov(e)\\p\in\ell'}} \frac{c_{\ell'}}{c_\ell} x_{\ell'} & \text{if } p \in V[G^u], q = u,
										\end{cases} \quad \text{for all } \ell \in L.
\]
$x^v$ is defined symmetrically. Notice that $\tfrac{c_{\ell'}}{c_\ell} \geq 1$: this term ensures that the cost in the sub-trees does not decrease -- this is important for the bundle constraints. Figure~\ref{figure:edgesplittingS} depicts this splitting. If $G^u$ contains an $\alpha$-thin edge with respect to $x^u$, we repeat this procedure for $G^u$ and $x^u$, and similarly for $G^v$ and $x^v$ if $G^v$ contains an $\alpha$-thin edge with respect to $x^v$. At the end, we obtain a set of pairs $\setc{(G^i,x^i)}{i \in [k]}$ such that no $G^i$ contains an $\alpha$-thin edge with respect to $x^i$. We refer to the contraction of $E^h$ and the repeated edge-splitting as \emph{$\alpha$-thin edge decomposition} of $G$ and $x$. Notice that we have $\supp(x^i) \cap \supp(x^j) = \emptyset$ for $i \neq j$ -- this allows~\cite{Adjiashvili17} and us to round each $(G^i,x^i)$ independently from each other later on. 

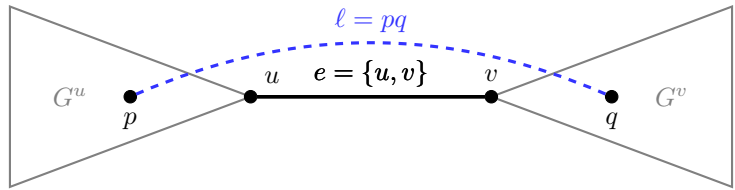
\begin{figure}[htbp]
 \centering
 \begin{tikzpicture}[xscale=0.8,yscale=0.6]
	\begin{scope}[xshift=0cm]
	 \draw[thick,gray] (0,0) -- (-4,-2) -- (-4,2) -- (0,0);
	 \node[gray] (gu) at (-3,0) {$G^u$};
	 \draw[thick,gray] (4,0) -- (8,-2) -- (8,2) -- (4,0);
	 \node[gray] (gu) at (7,0) {$G^v$};
   \node[node,fill,label={45:$u$}] (u) at (0,0) {};
	 \node[node,fill,label={90:$v$}] (v) at (4,0) {};
	 \draw[edge] (u) -- node[auto] {$e = \set{u,v}$} (v);
	
	 \draw[edge] (u) -- node[auto] {$e = \set{u,v}$} (v);
	 \draw[edge] (u) -- node[auto] {$e = \set{u,v}$} (v);
	 
	 \node[node,fill,label={270:$p$}] (p) at (-2,0) {};
	 \node[node,fill,label={270:$q$}] (q) at (6,0) {};
	 \draw[link] (p) edge[bend left] node[auto] {$\ell = pq$} (q);
	\end{scope}
 \end{tikzpicture}
 \caption{We split along $e = \set{u,v}$. For the dashed link $\ell = pq$ we have $x^u_\ell = x^v_\ell = 0$, since it covers $e$. In order to ensure that the coverage of edges in $G^u$ and $G^v$ does not decrease, we increase $x^u_{pu}$ and $x^v_{vq}$ by $\tfrac{c_\ell}{c_{pu}}x_\ell$ and $\tfrac{c_\ell}{c_{vq}}x_\ell$, respectively. \label{figure:edgesplittingS}}
\end{figure}

In the next lemma, we list important properties that the decomposition $(G^i,x^i), i \in [k]$ has -- these properties were shown in~\cite{Adjiashvili17} for the bundle \LP, but they also apply to the \oflp. One of the key properties is that the pairs $(G^i,x^i)$ are \emph{$\beta$-simple}. A pair $(G^i,x^i)$ is called \emph{$\beta$-simple} for a $\beta \in \N$, if there exists a \emph{$\beta$-center} $v \in V[G^i]$ such that removal of $v$ decomposes $G^i$ into a forest of $t$ trees $K_1,\dots, K_t$ such that for all $j \in [t]$:
\begin{enumerate}
 \item the cost of links in $K_j$ is bounded by $\beta$:
  \begin{equation}\label{eq:beta1}
	 \sum_{\ell \in L, \ell \subseteq V[K_j]} c_\ell x_\ell \leq \beta, 
	\end{equation}
 \item $K_j$ has at most $\beta$ leaves.
\end{enumerate}
Notice that the trees $K_1,\dots,K_t$ contain neither the $\beta$-center $v$ nor edges incident to $v$.

\begin{restatable}[Section 3.1 in~\cite{Adjiashvili17}]{lemma}{LemmaDecomposition}\label{lemma:decomposition}
 Let $(G,L,c)$ be a \WTAP instance with $c_\ell \leq M$ for all $\ell \in L$ and a constant $M \in \R_{\geq 1}$, and let $x \in \Rnn^L$ be a feasible solution to the \oflps for $\gamma = \ceil{\tfrac{28M}{\eps^2}}$. Let $E^h$ be the heavily-covered edges of $G$ with respect to $x$, $L^h$ be a set of links that covers $E^h$, $\overline{G}$ be the graph obtained from $G$ by contracting $L^h$ and $(G^i,x^i)$, $i \in [k]$ the $\tfrac{4M}{\eps^2}$-thin edge decomposition of $\overline{G}$ and $x$. Let $E^s$ be the set of edges along which we have split in the $\tfrac{4M}{\eps^2}$-thin edge decomposition. Then
 \begin{enumerate}
	\item $x^i$ is a feasible solution to the odd-cut \LP for $G^i$, and it fulfills
   \begin{equation}
    \label{eq:bundle-decomp}
    \sum_{\ell \in \cov(B)} c_\ell x_\ell \geq \OPT(B) \quad \text{ for all } B \in \mathcal{B}_\gamma^{G^i}
   \end{equation}
	 where $\mathcal{B}_\gamma^{G^i}$ is the collection of all edge sets in $E[G^i]$ that are a $\gamma$-bundle in $G$.
	\item $\supp(x^i) \cap \supp(x^j) = \emptyset$ for $i,j \in [k], i \neq j$,
  \item every $(G^i,x^i)$, $i \in [k]$ is \emph{$\tfrac{10M}{\eps^2}$-simple} and we have 
   \begin{equation}\label{eq:beta2}
	  \sum_{\ell \in L, \ell \subseteq V[K_j]} c_\ell x_\ell \leq \frac{4M}{\eps^2} 
	 \end{equation}	
	 for the trees $K_j, j\in [t]$ that are created by removing a $\beta$-center $v$ from $G^i$.	
	\item $\sum_{i \in [k]} c^\tp x^i \leq \left(1+\eps\right)c^\tp x$,
  \item we can efficiently compute sets of links $L^h, L^s \subseteq L$ covering $E^h, E^s$, respectively, with 
	  $c(L^h) \leq \eps c^\tp x$ and $c(L^s) \leq O(\eps) \sum_{i\in [k]} c^\tp x^i$.
 \end{enumerate}
\end{restatable}

The proof of this lemma can be found in Appendix~\ref{appendix:decomposition}, and is based on~\cite[Section 3.1]{Adjiashvili17}.
Properties 1 and 2 allow us to round each $x^i$ individually, and property 3 makes it possible to employ a rounding algorithm from~\cite[Lemma 3.8]{Adjiashvili17}. Properties 4 and 5 ensure that the decomposition costs us only a factor of $(1+O(\eps))$, and that the edges not contained in a $G^i$ can also be covered at cost $O(\eps)\OPT$. Thus, the approximation guarantee is dominated by how well we can round the individual $x^i$ solutions. 

\section{Rounding the Solution\label{sec:rounding}}

Given a \WTAP instance $(G,L,c)$, the overall algorithm begins by solving the \oflps for $\gamma = \ceil{\tfrac{28M}{\eps^2}}$. It then computes a solution for $L^h$ for the heavily-covered edges $E^h$, contracts $L^h$ to obtain a new graph $\overline{G}$ and decomposes $\overline{G}$ and the solution $x$ into $\tfrac{10M}{\eps^2}$-simple pairs $(G^i,x^i), i \in [k]$, as described in Section~\ref{sec:decomposition}. In the process, it also computes a set of links $L^s$ covering $E^s$ as defined in Lemma~\ref{lemma:decomposition}.

Afterwards, we round each $(G^i,x^i)$ individually using the best of the following lemmas -- the first one performs well for pairs where cross-links carry much of the cost and is based on the integrality of the odd-cut \LP for instances with only cross- and up-links, while the second performs well for pairs where in-links carry much of the cost and is due to~\cite{Adjiashvili17}; it relies on the bundle constraints. The output of the algorithm is then the union of $L^h$ and $L^s$ as computed by Lemma~\ref{lemma:decomposition}, with the rounded solutions for all $(G^i,x^i)$.

\begin{lemma}
 \label{lemma:rounding}
 Let $x$ be a feasible solution to the odd-cut \LP for $G$. Let $G$ be rooted at any node $r \in V[G]$, then we can compute in polynomial time a solution $S \subseteq L$ covering $E[G]$ with cost at most
 \begin{equation}
  c(S) \leq 2c^\tp \xin + c^\tp \xcr.
 \end{equation} 
\end{lemma}

\begin{proof}
 We define a new feasible solution for the odd-cut \LP for $G$ by the applying to the following modification to all in-links with mass that are not up-links, \ie all $\ell \in \supp(x) \cap (\Lin \setminus \Lup)$. Let $\ell = uv$ and $w := \lca(uv)$. Then we increase $x_{uw}$ and $x_{wv}$ by $x_\ell$, and set $x_\ell := 0$, and we call the resulting solution $y$. 

 By construction, we have $\supp(x) \cap (\Lin \setminus \Lup) = \emptyset$ -- that means that the support of $y$ contains only up- and cross-links. At the same time, we know that $c^\tp y \leq 2c^\tp \xin + c^\tp\xcr$ since we doubled the mass of in-links. Notice finally that $y$ is feasible for the odd-cut \LP, since the total mass of links covering each edge did not change by the modification.

 Now resolve the odd-cut \LP with the set of links restricted to $\supp(y)$. We know by Theorem~\ref{theorem:exact} that the odd-cut \LP is integer for this set of links, allowing us to obtain an integral solution $z$ with $c^\tp z \leq c^\tp y \leq 2c^\tp \xin + c^\tp\xcr$.
\end{proof}

The idea of the rounding algorithm from~\cite{Adjiashvili17} is as follows. Given a $\tfrac{10M}{\eps^2}$-simple pair $(H,x)$ with a $\beta$-center $r$, we replace all cross-links $\ell =uv$ by two up-links $ur$ and $rv$. After that, we split $H$ at $r$ into trees $H_1,\dots, H_m$ such that $r$ is part of all trees. Since we no longer have cross-links, we can split $x$ into feasible solutions $x^1, \dots, x^m$ for the bundle \LP for $H_1,\dots,H_m$ with $\supp(x^i) \cap \supp(x^j) = \emptyset$ if $i \neq j$. Every sub-tree $H_1, \dots, H_m$ of $H$ has at most $\tfrac{10M}{\eps^2}+1$ leaves (the root causes the additional $+1$); it can then be shown that every $H_j$ was a $\ceil{\tfrac{28M}{\eps^2}}$-bundle in the original graph, allowing us to use a bundle constraint to bound the rounding cost.

\begin{restatable}[Lemma 3.8 in~\cite{Adjiashvili17}]{lemma}{RoundingDavid}\label{lemma:david}
 Let $x$ be a feasible solution to the \oflps for $G$, with $\gamma = \ceil{\tfrac{28M}{\eps^2}}$.
 Given a $\frac{10M}{\eps^2}$-simple pair $(H,x)$ from the $\tfrac{4M}{\eps^2}$-thin decomposition of $G$ and $x$ that is rooted at a  $\frac{10M}{\eps^2}$-center $r$, we can compute in polynomial time a set of links $S \subseteq L$ covering $E[H]$ with cost at most
 \begin{equation}
  c(S) \leq c^\tp \xin + 2c^\tp \xcr + |V^h \cap V[H]|
 \end{equation}
 with $V^h$ being the set of nodes created by the contractions of the links in $L^h$.
\end{restatable}

For completeness, we give a proof of this lemma in Appendix~\ref{appendix:rounding}. Together, these two lemmas allow us to proof the following theorem.

\OnePointHalf*

\begin{proof}
 For constant $\gamma$, the \oflps can be solved in polynomial time due to Lemma~\ref{lem:solving_main_lp}. The decomposition can be performed in polynomial time, as can the computation of $L^h$ and $L^s$ in Lemma~\ref{lemma:decomposition} and the rounding schemes of Lemmas~\ref{lemma:rounding} and \ref{lemma:david}. Since we cover all $E[G^i]$ as well as $E^h$ and $E^s$, we produce a feasible solution due to Observation~\ref{obs:feasibility}.

 What is left to analyze the approximation ratio of the algorithm. By Lemma~\ref{lemma:decomposition}, we know that $c(L^h \cup L^s) \leq O(\eps) \cdot \OPT$. For a pair $(H,x)$ in the decomposition, we know by Lemmas~\ref{lemma:rounding} and~\ref{lemma:david} that the cost for its covering $S$ can be bounded by
\[
 \begin{aligned}
  c(S) &\leq \min(2c^\tp \xin + c^\tp \xcr,c^\tp \xin + 2c^\tp \xcr + |V^h \cap V[H]|)\\
     	 &\leq \frac{1}{2}\left(2c^\tp \xin + c^\tp \xcr + c^\tp \xin + 2c^\tp \xcr + |V^h \cap V[H]|\right)\\
			 &\leq \frac{3}{2} c^\tp \xin + \frac{3}{2} c^\tp \xcr + \frac{1}{2} |V^h \cap V[H]|.\\
 \end{aligned}				
\]
Thus, the cost of the computed coverings $S_1,\dots,S_k$ for $(G^1,x^1), \dots, (G^k,x^k)$ is bounded by
\[
 \begin{aligned}
  \sum_{i \in [k]} c(S_i) \leq |V^h| + \sum_{i \in [k]} \frac{3}{2}c^\tp x^i = |V^h| + \frac{3}{2} c^\tp x.
 \end{aligned}				
\]
Since we have $|V^h| \leq |L^h| \leq c(L^h) \in O(\eps) \OPT$, the claim follows (every link costs at least $1$, and we need at least one link for a node in $V^h$). 
\end{proof}

\paragraph*{Acknowledgements.}

We thank Nikos Mutsanas for early discussions on the odd-cut \LP, L\'aszl\'o V\'egh and Giacomo Zambelli for their feedback.

\bibliography{literature}

\appendix
\section{Proof of Lemma~\ref{lemma:decomposition}\label{appendix:decomposition}}

In this section, we prove that the properties of the decomposition in Lemma~\ref{lemma:decomposition} also hold if we use the \oflp as a basis for the decomposition instead of the bundle \LP. The proof ideas are all due to~\cite{Adjiashvili17}. Notice that we use $\gamma = \ceil{\frac{28M}{\eps^2}}$ and $\beta = \frac{10M}{\eps^2}$ instead of $\gamma = \ceil{\frac{56M}{\eps^2}}$ and $\beta = \frac{12M}{\eps^2}$ in~\cite{Adjiashvili17}; this is the result of slightly improved estimations. In particular, we use that:
\begin{itemize}
 \item Links covering an edge $e = \set{u,v}$ can cover at most one leaf in the subtrees $G^u$ and $G^v$. This argument allows us to reduce $\beta$ from $\frac{12M}{\eps^2}$ to $\frac{10M}{\eps^2}$.
 \item The cost of links completely in a sub-tree of a $\tfrac{10M}{\eps^2}$-simple pair is actually bounded by $\tfrac{4M}{\eps^2}$, even though the pair might not be $\tfrac{4M}{\eps^2}$-simple due to the number of leaves of the sub-trees.
 \item The previous point allows us bound $|V^h|$ by $\tfrac{6M}{\eps^2}$; together with the improved $\beta = \frac{10M}{\eps^2}$ we obtain $\gamma = \ceil{\frac{28M}{\eps^2}}$.
\end{itemize}
Details can be found in the proof below.

\LemmaDecomposition*

\begin{proof} We now prove the five properties in the order in which they were stated.

\begin{enumerate}


 \item {\bf Feasibility}: We start with a solution $x$ that is feasible for the \oflp for $G$, then we contract the links in $L^h$ to obtain a new graph $\overline{G}$. Since we do not modify the solution during the contraction, $x$ fulfills \eqref{eq:bundle-decomp} for $B \in \mathcal{B}^{\overline{G}}_\gamma$ by definition of $\mathcal{B}^{\overline{G}}_\gamma$, because $\mathcal{B}_\gamma^{G'}$ is the set of edges in $G'$ that are $\gamma$-bundles in $G$.

 The odd-cut constraints for $E[\overline{G}]$ are
 \begin{equation*}
 \sum_{\ell \in L} \ceil{\tfrac{1}{2}|P_\ell \cap \delta_{\overline{G}}(S)|} x_\ell \geq \frac{|\delta_{\overline{G}}(S)|+1}{2} \qquad \text{for all } S \in \mathcal{S}_{\overline{G}}
\end{equation*}
with $\mathcal{S}_{\overline{G}}$ being the collection of all sets $S$ such that $|\delta_{\overline{G}}(S)|$ is odd. Let $\mathcal{S}_{G}$ being the collection of all sets $S$ such that $|\delta_{G}(S)|$ is odd. 
For a set $S \in \mathcal{S}_{\overline{G}}$, let $S' \in \mathcal{S}_{G}$ be the set with $\delta_{\overline{G}}(S) = \delta_G(S')$. The existence of $S'$ is guaranteed since $\overline{G}$ is a contraction of the tree $G$. Since we do not modify the solution during the contraction, the existence of $S'$ for every $S$ guarantees that the odd-cut constraints are fulfilled.

Now assume that we have a graph $H$ and a solution $x$ which is feasible for the \oflp for $G$, and fulfills \eqref{eq:bundle-decomp} for all $B \in \mathcal{B}^H_\gamma$. We now show that splitting an edge $e = \set{u,v} \in E[H]$ results in two pairs $(H^u,x^u)$, $(H^v,x^v)$ such that $x^u$, $x^v$ are feasible for the \oflp for $H^u$, $H^v$ and fulfill \eqref{eq:bundle-decomp} for all $B \in \mathcal{B}^{H^u}_\gamma$ and $B \in \mathcal{B}^{H^v}_\gamma$, respectively.

Recall that splitting an edge $e = \set{u,v}$ results in $x^u$ (and analogously $x^v$) in the following way:
\[
 x^u_{\ell=pq}\ := \begin{cases} 
                     x_\ell                                      & \text{if } p,q \in V[G^u] \setminus\set{u},\\ 
										 0                                           & \text{if } p \in V[G^v]\text{ or }q \in V[G^v], \\ 
										 x_\ell + \displaystyle\sum\limits_{\substack{\ell' \in \cov(e)\\p\in\ell'}} \frac{c_{\ell'}}{c_\ell} x_{\ell'} & \text{if } p \in V[G^u], q = u,
										\end{cases}  \quad \text{for all } \ell \in L.
\]
This definition differs from~\cite{Adjiashvili17} in the additional $\frac{c_{\ell'}}{c_\ell}$ term that ensures that bundle constraints remain feasible in the case that the shadows to which mass is shifted to are cheaper than the links where the mass originates from.
Due to this, we have that for every edge $e \in E[G^u]$
\[ \sum_{\ell \in \cov(e)} x^u_\ell \geq \sum_{\ell \in \cov(e)} x_\ell, \quad \text{and } \sum_{\ell \in \cov(e)} c_\ell x^u_\ell \geq \sum_{\ell \in \cov(e)} c_\ell x_\ell \]
 by definition of the splitting; the same applies to any edge set $F \subseteq E[G^u]$:
\[  \sum_{\ell \in \cov(F)} x^u_\ell \geq \sum_{\ell \in \cov(F)} x_\ell, \quad \text{and}\quad  \sum_{\ell \in \cov(F)} c_\ell x^u_\ell \geq \sum_{\ell \in \cov(F)} c_\ell x_\ell. \]
This implies that the left hand sides of \eqref{eq:bundle-decomp} can only get larger in $G^u$, while the right hand sides stay the same. 

For the odd-cut constraints, consider a set $S \subseteq V[G^u]$ such that $|\delta_{G^u}(S)|$ is odd, and let $\mathcal{S}_{G^u}$ be the collection of these sets. The important fact is here that the mass of links covering an edge in $E[G^u]$ does not decrease. We define $\alpha_{\ell,S} := \ceil{\tfrac{1}{2}|P_\ell \cap \delta_{G^u}(S)|}$ for all $\ell \in L, S \in \mathcal{S}_{G^u}$, and note that $\alpha_{\ell,S} = \alpha_{\ell',S}$ for $\ell = pu \in L$, $\ell' = pq \in \cov(e)$ yielding
 \[\begin{aligned}
  &\  \sum_{\ell \in L} \ceil{\tfrac{1}{2}|P_\ell \cap \delta_{G^u}(S)|} x^u_\ell \\[-4mm]
 =&\  \sum_{\ell \in L, \ell \subseteq V[G^u] \setminus\set{u}} \alpha_{\ell,S}\cdot x_\ell + \sum_{\ell = pu \in L} \alpha_{\ell,S} \left(x_\ell+\sum_{\substack{\ell' \in\cov(e)\\p\in\ell'}} \frac{c_{\ell'}}{c_\ell}x_{\ell'} \right) + \sum_{\ell \in L, \ell \in \cov(e)} \alpha_{\ell,S}  \cdot 0\\
 \geq&\  \sum_{\ell \in L, \ell \subseteq V[G^u] \setminus\set{u}} \alpha_{\ell,S}\cdot x_\ell + \sum_{\ell = pu \in L} \alpha_{\ell,S} x_\ell + \sum_{\ell = pu \in L} \sum_{\substack{\ell' \in\cov(e)\\p\in\ell'}} \alpha_{\ell',S}\cdot x_{\ell'}\\
 \geq&\  \sum_{\ell \in L, \ell \subseteq V[G^u] \setminus\set{u}} \alpha_{\ell,S}\cdot x_\ell + \sum_{\ell = pu \in L} \alpha_{\ell,S} x_\ell + \sum_{p \in V[G^u]} \sum_{\ell = pq \in\cov(e)} \alpha_{\ell,S}\cdot x_{\ell}\\
 \geq&\  \sum_{\ell \in L, \ell \subseteq V[G^u] \setminus\set{u}} \alpha_{\ell,S}\cdot x_\ell + \sum_{\ell = pu \in L} \alpha_{\ell,S} x_\ell + \sum_{\ell \in L, \ell \in \cov(e)} \alpha_{\ell,S}\cdot x_{\ell}\\
\geq&\ \frac{|\delta_{G^u}(S)|+1}{2} \qquad \text{for all } S \in \mathcal{S}_{G^u}.
\end{aligned}
\]
This proves adherence to the odd-cut constraints.
Analogous statements apply to $G^v$ and $x^v$. 
This argumentation can now be used inductively on the splitting process -- this completes the proof of this property.

  
 \item {\bf Disjointness}: By definition of the edge splitting, we know that $V[G^i]$ and $V[G^j]$ are disjoint for all $i \neq j, i,j \in [k]$. Furthermore, for any link $\ell \in \supp(x^i)$, we have $\ell \subseteq V[G^i]$, proving this claim. 
 

 \item {\bf Simplicity}: Here, we have to show that every $(G^i,x^i)$ is $\tfrac{10M}{\eps^2}$-simple, which requires that there is a node $v$ for every $G^i$ whose removal decomposes $G^i$ into trees $K_1,\dots,K_t$ with two properties: 
  \begin{itemize}
	 \item $K_j$ fulfills $\sum_{\ell \in L, \ell \subseteq V[K_j]} c_\ell x_\ell \leq \frac{10M}{\eps^2}$ for all $j \in [t]$,
	 \item $K_j$ has at most $\frac{10M}{\eps^2}$ leaves for all $j \in [t]$.
  \end{itemize}
	Furthermore, we have to prove that
   \begin{equation*}
	  \sum_{\ell \in L, \ell \subseteq V[K_j]} c_\ell x_\ell \leq \frac{4M}{\eps^2} \quad \text{for all } j \in [t].
	 \end{equation*}	
	 Notice	that this implies the first property of $(G^i,x^i)$ being $\tfrac{10M}{\eps^2}$-simple.
	
  Let $(H,x) := (G^i,x^i)$ for some $i \in [k]$. We know that every edge $e = \set{u,v} \in E[H]$ satisfies 
  \begin{equation}\label{eq:alphaviol}  \sum_{\ell \in L, \ell \subseteq V[H^u]} c_\ell x_\ell < \frac{4M}{\eps^2} \quad \text{ and / or } \quad \sum_{\ell \in L, \ell \subseteq V[H^v]} c_\ell x_\ell   < \frac{4M}{\eps^2}. \end{equation}
  We distinguish two cases:
  \begin{enumerate}
   \item In this case, there exists an edge $e = \set{u,v} \in E[H]$ with:
	  \begin{equation*}\label{eq:decompcase1}
     \sum_{\ell \in L, \ell \subseteq V[H^u]} c_\ell x_\ell < \frac{4M}{\eps^2} \quad \text{  and  }\quad  \sum_{\ell \in L, \ell \subseteq V[H^v]} c_\ell x_\ell < \frac{4M}{\eps^2}.
    \end{equation*}
	  Thus, choosing $u$ or $v$ decomposes $H$ into subtrees $K_1,\dots,K_t$ with 
	  \begin{equation*}\label{eq:decompcase1b}
     \sum_{\ell \in L, \ell \subseteq V[K_i]} c_\ell x_\ell < \frac{4M}{\eps^2}
    \end{equation*}
	  for all $i \in [t]$.
	 \item Otherwise, we have for all edges $e = \set{u,v} \in E[H]$, that 
	  \begin{equation}
		 \text{either } \sum_{\ell \in L, \ell \subseteq V[H^u]} c_\ell x_\ell < \frac{4M}{\eps^2} \quad \text{  or } \quad \sum_{\ell \in L, \ell \subseteq V[H^v]} c_\ell x_\ell   < \frac{4M}{\eps^2}. 
		\end{equation}
	  We orient every edge $e = \set{u,v}$ based on whether $\sum_{\ell \in L, \ell \subseteq V[H^u]} c_\ell x_\ell < \frac{4M}{\eps^2}$ holds: if yes, we orient it from $u$ to $v$, otherwise from $v$ to $u$. Since $H$ is a tree, there is a node $v \in V[H]$ with $\outdeg(v) = 0$. Removing $v$ again decomposes $H$ into sub-trees $K_1,\dots,K_t$ with 
	 \begin{equation*}\label{eq:decompcase1c}
    \sum_{\ell \in L, \ell \subseteq V[K_i]} c_\ell x_\ell < \frac{4M}{\eps^2}
   \end{equation*}
	 for all $i \in [t]$.
 \end{enumerate}
 This completes the first part of the proof. Now we bound the number of leaves for a tree $K_j, j \in [t]$. 

 Now consider a tree $K_j$, $j \in [t]$, and let $e$ be the edge that connected $K_j$ to the $\beta$-center in $G^i$ that was removed. The number of leaves that a tree $K_j$ has can be bounded by the following argument. Every link $\ell = uv$ has cost at least 1 and can contribute to the covering of at most 2 leaf edges of $K_j$ if $u,v \in V[K_j]$; if either $u$ or $v$ is not in $V[K_j]$ then $\ell$ can cover at most one of $K_j$ leaf edge. Finally, if both $u,v \not\in V[K_j]$ it cannot cover leaf edges of $K_j$ at all. 

 We get the following bound for the number of leaves $|\leaves(K_j)|$ of $K_j$ (since $e \not\in E^h$):
 \[ |\leaves(K_j)| \leq 2\cdot\sum_{\ell \in L, \ell \subseteq V[K_j]} x_\ell + \sum_{\ell \in \cov(e)} x_\ell \leq 2\cdot\frac{4M}{\eps^2} + \frac{2}{\eps} \leq \frac{10M}{\eps^2}. \]
 This completes the proof of this property.


 \item {\bf Cost Increase}: We have to show that 
  \[  \sum_{i \in [k]} c^\tp x^i \leq \left(1+\eps\right)c^\tp x. \]
	Consider an edge-splitting of a pair $(H,x) := (G^i,x^i)$ along an edge $e = \set{u,v}$ into $H^u, H^v, x^u$ and $x^v$. We have $c^\tp x \leq c^\tp x^u + c^\tp x^v + \sum_{\ell \in \cov(e)} c_\ell x_\ell$ by definition of the edge-splitting. Thus, each splitting increases the cost by $\sum_{\ell \in \cov(e)} c_\ell x_\ell < \tfrac{2}{\eps}M$ due to our link cost bound $M$ and the fact that all edges with a mass of more than $\tfrac{2}{\eps}$ were contracted. Since we end up with $k$ pairs, we have $k-1$ splits and the total increase in cost is bounded by $(k-1)\tfrac{2M}{\eps}$:
 \[
   \sum_{i\in [k]} c^\tp x^i  \leq c^\tp x + (k-1)\frac{2M}{\eps} \leq c^\tp x + k\frac{2M}{\eps}.
 \]
 On the other hand, we split along $\tfrac{4M}{\eps^2}$-thin edges and have therefore
 \[
  c^\tp x^i = \sum_{\ell \in L, \ell \subseteq V[G^i]} c_\ell x_\ell \geq \frac{4M}{\eps^2} \quad \text{ for all } i \in [k].
 \]
 Combined, we now have for $\eps \leq 1$:
 \[
  \begin{aligned}
   \sum_{i\in [k]} c^\tp x^i 
	  \leq c^\tp x + k\frac{2M}{\eps} 
	 &= \left(1+\frac{k\frac{2M}{\eps}}{c^\tp x}\right) c^\tp x\\
	 &\leq \left(1+\frac{k\frac{2M}{\eps}}{ \sum_{i\in [k]} c^\tp x^i-k\frac{2M}{\eps}}\right) c^\tp x\\
	 &\leq \left(1+\frac{k\frac{2M}{\eps}}{k\frac{4M}{\eps^2}-k\frac{2M}{\eps}}\right) c^\tp x\\
	 &= \left(1+\frac{1}{\frac{2}{\eps}-1}\right) c^\tp x = \left(1+\frac{\eps}{2-\eps}\right) c^\tp x \leq \left(1+\eps\right) c^\tp x.\\
	\end{aligned}
 \]
 That completes this part of the proof.


 \item {\bf Remaining Edges}: We need to show the existence of link sets $L^h, L^s$ that cover $E^h$ and $E^s$, respectively, with $c(L^h) \leq \eps c^\tp x$ and $c(L^s) \leq O(\eps) \sum_{i\in [k]} c^\tp x^i$.

  Consider the graph $G'$ obtained by contracting all edges in $E[G] \setminus E^h$. Since $E^h := \setc{e \in E[G]}{x(\cov(e)) \geq \frac{2}{\eps}}$ we know that $y := \tfrac{\eps}{2}x$ is a feasible solution to the cut \LP for $G'$. The cut \LP is known to have an \LP gap of at most 2 due to various 2-approximation algorithms, \eg~\cite{GoemansEtAl94,Jain01} or \cite[Proposition 2.1]{Adjiashvili17}; applying this to $y$ yields  a set of links $L^h$ with 
	\[ c(L^h) \leq 2 c^\tp y = 2\tfrac{\eps}{2} c^\tp x = \eps c^\tp x. \] 
	
	This shows the first part of the claim. For the second part, remember that all edges $e \in E^s$ used for splitting were $\tfrac{4M}{\eps^2}$-thin edges, \ie 
	\[ c^\tp x^i = \sum_{\ell \in L, \ell \subseteq V[G^i]} c_\ell x_\ell \geq \frac{4M}{\eps^2} \quad \text{for all } i\in [k]. \]
	Since $|E^s| = k-1$, any inclusion-wise minimal covering of $E^s$ costs at most $(k-1)M$ and we have
  \[
   (k-1)M = \frac{(k-1)M}{\sum_{i\in [k]} c^\tp x^i}\sum_{i\in [k]} c^\tp x^i \leq \frac{(k-1)M}{k\tfrac{4M}{\eps^2}}\sum_{i\in [k]} c^\tp x^i \leq \frac{\eps^2}{4}\sum_{i\in [k]} c^\tp x^i \in O(\eps) \sum_{i\in [k]} c^\tp x^i .
  \]
 \end{enumerate}
\end{proof}

\section{Proof of Lemma~\ref{lemma:david}\label{appendix:rounding}}

\RoundingDavid*

\begin{proof}
 Consider the following rounding procedure. Given $(H,x)$, let $r$ be a $\tfrac{10M}{\eps^2}$-center of $(H,x)$. Let $H_1,\dots,H_m$ be the sub-trees of $H$ created by removing $r$. For a sub-tree $H_i$, let $e_i$ be the edge that connects $H_i$ to $r$ in $H$. We define new sub-trees $\overline{H}_i, i \in [m]$ by adding $e_i$ and $r$ to $H_i$. We now create a new solution $y$ that contains no cross-links by splitting all cross-links at $r$ into two up-link shadows by defining:
 \begin{equation*}
  y_{\ell=uv} := \begin{cases} x_\ell & \ell \in \supp(\xin), r \not\in \ell\\ x_\ell + \sum_{\substack{\ell' \in \supp(\xcr)\\v \in \ell'}} x_{\ell'} & u=r,v\in V[H]\setminus\set{r} \\ 0 & \text{else} \end{cases} \quad \text{ for all } \ell \in L.
 \end{equation*}
 Notice that we have $c^\tp y \leq c^\tp\xin + 2c^\tp \xcr$.
 Since $y$ contains no cross-links, it is the union of $m$ disjoint solutions, one for each $\overline{H}_i, i \in [m]$. We will now focus on rounding the solution $z$ for a tree $\overline{H} \in \setc{\overline{H}_i}{ i \in [m]}$ with $z$ being defined by
 \begin{equation*}
  z_{\ell} := \begin{cases} y_\ell & \ell \subseteq V[\overline{H}]\\ 0 & \text{else} \end{cases} \quad \text{ for all } \ell \in L.
 \end{equation*}
 We now make the following claim.

 \begin{claim}\label{cl:compound}
  If $z$ is a feasible solution to the bundle \LP 
  \begin{equation*}
	 \OPT(E[\overline{H}]) \leq c^\tp z + |V^h \cap V[\overline{H}]|.
	\end{equation*}
 \end{claim}

 \begin{proof}
  Let $e$ be the edge connecting $\overline{H}$ to $r$ in $H$.
  Since $(H,x)$ is $\frac{10M}{\eps^2}$-simple and $e \not\in E^h$, we know that $c^\tp z \leq \tfrac{10M}{\eps^2} + x(\cov(e)) M \leq \frac{12M}{\eps^2}$. Due to Lemma~\ref{lemma:decomposition}, property 3 we even have that $c^\tp z \leq \tfrac{4M}{\eps^2} + x(\cov(e)) M \leq \frac{6M}{\eps^2}$.
	
	Since $z$ fulfills the cut constraints, we know that $\OPT(E[\overline{H}]) - c^\tp z \leq c^\tp z \leq \frac{6M}{\eps^2}$. Thus, if $|V^h \cap V[\overline{H}]| \geq \frac{6M}{\eps^2}$, we are done. We therefore assume that $|V^h \cap V[\overline{H}]| < \frac{6M}{\eps^2}$.
	In this case, we want to employ the bundle constraints that $z$ fulfills, meaning that we have to analyze whether $\overline{H}$ is a $\ceil{\tfrac{28M}{\eps^2}}$-bundle in $G$.
	
	Due to $(H,x)$ being $\frac{10M}{\eps^2}$-simple, we know that the number of leaves of $\overline{H}$ is bounded by $\tfrac{10M}{\eps^2}+1$. We now split $\overline{H}$ at all nodes with degree at least 3. The number of nodes with degree at least 3 is bounded by the number of leaves. As a consequence, we obtain a decomposition into $Q_1, \dots, Q_t$ paths with $t \leq 2 \cdot \left(\tfrac{10M}{\eps^2}+1\right)$.
	
	Each $Q_i$, $i\in [t]$ is the union of paths in $G$, separated by components that were contracted into a node of $V^h$.  If a $Q_i$ is the union of $k$ paths, there need to be $k-1$ nodes in the interior of the path that by construction cannot be part of other $Q_i$, since $Q_i$ has no degree 3 or higher nodes in its interior. Thus, each node of $V^h$ is only capable of splitting a single path in $\overline{H}$ into two paths in $G$. Thus, we have at most $2 \cdot \left(\tfrac{10M}{\eps^2}+1\right)$ paths $Q_1, \dots, Q_t$ in $\overline{G}$ that are splitted into two paths in $G$ at most $\tfrac{6M}{\eps^2}$ times. Therefore, $\overline{H}$ is the union of at most $2 \cdot \left(\tfrac{10M}{\eps^2}+1\right) + \tfrac{6M}{\eps^2} \leq \tfrac{28M}{\eps^2}$ paths in $G$ and is therefore a $\ceil{\tfrac{28M}{\eps^2}}$-bundle in $G$.
	
	Thus, we have by the corresponding bundle-constraint:
	\[
   c^\tp z = \sum_{\ell \in \cov(\overline{H})}	c_\ell z_\ell \geq \sum_{\ell \in \cov(\overline{H})}	c_\ell x_\ell \geq \OPT(E[\overline{H}])
	\]
	which completes the proof of this claim.
 \end{proof}
 
 Claim~\ref{cl:compound} completes the proof, except for the run-time. However, we can use Lemma~\ref{lemma:constantleaves} to compute solutions for $\gamma$-bundles, and since $\OPT(E[\overline{H}]) \in O(M)$, we can use simple enumeration in the case of many contracted nodes.
 
\end{proof}

\end{document}